\documentclass[review]{elsarticle}

\usepackage[left=3cm, right=3cm, top=3cm]{geometry}
\usepackage{lineno,hyperref,amsmath,amsthm,esvect}

\modulolinenumbers[5]

\journal{International Journal of Computer Mathematics: Computer Systems Theory}

\newtheorem{theorem}{Theorem}[section]
\newtheorem{lemma}{Lemma}[section]

\graphicspath{ {Figures/} }

\begin{document}

\begin{frontmatter}

\title{\bf\large A sufficient condition for visibility paths in simple polygons}

\author[mymainaddress]{Mohammad Reza Zarrabi\corref{mycorrespondingauthor}}
\cortext[mycorrespondingauthor]{Corresponding author}
\ead{m.zarabi@modares.ac.ir}

\author[mymainaddress]{Nasrollah Moghaddam Charkari}
\ead{charkari@modares.ac.ir}
\address[mymainaddress]{Faculty of Electrical Engineering and Computer Science, Tarbiat Modares University, Tehran, Iran}

\begin{abstract}
The purpose of this paper is to give a simple proof for a necessary and sufficient condition for visibility paths
in simple polygons. A visibility path is a curve such that every point inside
a simple polygon is visible from at least one point on the path.
This result is essential for finding the shortest watchman route inside a simple polygon
specially when the route is restricted to curved paths.
\end{abstract}


\end{frontmatter}

\section{Introduction}
Visibility coverage of simple polygons with a mobile guard (mainly known as \emph{watchman} problems)
is a central problem in computational geometry.
Usually, a mobile guard is defined as a moving point inside a simple polygon $P$
that sees in any direction for any distance.
Also, visibility is defined by the clearance of straight line between two points inside $P$.
In other words, two points inside $P$ see each other if the line segment connecting them does not intersect
the boundary of $P$.
Consider the curved path created by the mobile guard during its walk inside $P$.
Coverage is achieved if every point inside $P$ is visible from at least one point on the path.
In this case, the path is called a \emph{visibility path}.

The \emph{watchman route problem} (WRP) asks for finding a visibility route
(a visibility path whose starting and end points coincide)
inside $P$ of minimum length.
The WRP is defined for two versions. The \emph{anchored} WRP, in which the tour is required to pass through a
specified anchor point, and the \emph{floating} WRP, in which no anchor point is specified.
The problem has a polynomial time solution for both versions.
The anchored WRP was first studied by Chin and Ntafos \cite {Chin_1991}.
Afterwards, many researchers have worked on this problem in
\cite {Tan_Hirata_Inagaki_1993}, \cite {Tan_Hirata_1993}, \cite {Hammar_1997}, \cite {Tan_Hirata_Inagaki_1999} and
\cite {Dror_2003}, respectively, to improve the results.
The floating WRP was investigated in \cite {Carlsson_1993, Dror_2003}.
The best known results for both versions were presented by Dror et al. \cite {Dror_2003}.
Linear time algorithms are known for approximating
the anchored WRP \cite {Tan_2004} and floating WRP \cite {Tan_2007}
(also, a non-linear time algorithm in \cite {Nilsson_2001}).
All the above solutions for the WRP used the following theorem (special case of the theorem for visibility routes):

\begin{theorem}\label{Main}
A curved path inside $P$ is a visibility path, iff it intersects all nonredundant cuts of $P$.
\end{theorem}

Although this theorem is essential for the WRP,
only its necessary condition is discussed in some of the above work
(the theorem was first mentioned in \cite {Chin_1991} without proof,
other work mentioned and used this theorem, but again without proving it), i.e.,
any visibility paths must intersect all nonredundant cuts, since otherwise,
some corners of $P$ will not be seen by the path.
On the other hand, this theorem cannot directly be obtained from the definition of nonredundant cuts in \cite {Das_1997}
(there are some vertices that might not be seen by the points on the path intersecting nonredundant cuts).
Thus, in this note, we prove that it is indeed also a sufficient condition.
Theorem~\ref{Main} would be helpful for variants of the WRP when the route is restricted to curved paths.
For example, a piece-wise route consisting of \emph{conic} sections (Figure~\ref{fig:Nonredundant}).
In the standard WRP, we are looking for the shortest \emph{polygonal} visibility route.

\begin{figure}
	\centering
	\includegraphics[width=7cm,keepaspectratio=true]{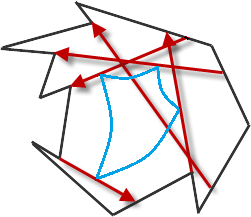}
	\caption{A conic route \cite {Garcia_2000} (blue) and nonredundant cuts (red)}
	\label{fig:Nonredundant}
\end{figure}

\section{Preliminaries}
Let $P$ be a simple polygon and let $\partial P$ denote the boundary of $P$.
Without loss of generality, assume that the vertices of $P$ are sorted in \emph{clockwise} order and
the whole $P$ is not visible from any interior point, i.e., $P$ is not a \emph{star-shaped} polygon.
A vertex of $P$ is \emph{reflex} or \emph{convex} depending on whether its
internal angle is strictly larger than $\pi$ or not, respectively.

Let $v$ be a reflex vertex, $u$ a vertex adjacent to $v$, and $w$ the closest point to $v$
on $\partial P$ hit by the half line originating at $v$ along $\overline{uv}$.
We call the line segment $C=\overline{vw}$ a \emph{cut} of $P$ \emph{generated} by $\overline{uv}$
(see \cite {Alsuwaiyel_1995, Zarrabi_2020}).
Each reflex vertex gives rise to two cuts.
Since $C$ partitions the polygon into two portions, we define $Pocket(C)$ to be the portion of $P$ that includes $u$
($\overline{uv}$ is the \emph{generator} of $Pocket(C=\overline{vw})$).
We also associate a direction to each cut $C$ such that
$Pocket(C)$ entirely lies on the right of $C$ (right hand rule). This direction is compatible with the
clockwise ordering of vertices of $P$.

If $C$ and $C'$ are two cuts such that $Pocket(C) \subset Pocket(C')$,
then $C'$ is called \emph{redundant}, otherwise it is \emph{nonredundant} \cite {Das_1997}.
The nonredundant cuts are termed the \emph{essential} cuts in \cite {Tan_2007}, where
a well-known linear time algorithm for computing essential cuts was presented.
We will be working mainly with the nonredundant cuts of $P$
(see Figure~\ref{fig:Nonredundant}).

A \emph{visibility path} for a given $P$ is
a connected path (possibly curved) contained in $P$ with the property that
every point inside $P$ is visible from at least one point on the path
(remember that as a continuous image of [0,1], a curved path is connected).

We define $\pi_E(x,y)$ to be the shortest Euclidean path from a point $x$ to a point $y$ inside $P$.
Similarly, $\pi_E(x,\Pi)$ is defined as a shortest Euclidean path from a point $x$ to a curved path $\Pi$ inside $P$.
Note that $\pi_E(x,\Pi)$ might not be unique, but if $x$ is not visible from $\Pi$, the first segment of $\pi_E(x,\Pi)$
will be unique.

\section{Main result}
In this section, we prove the sufficient condition of Theorem~\ref{Main}.
Let $\Pi$ be a curved path inside $P$ such that it intersects all nonredundant cuts of $P$.
We will show that $\Pi$ is a visibility path:

\begin{proof}

\begin{lemma}\label{Lemma_One}
If a curved path $\Pi$ inside $P$ sees all of $\partial P$, then it sees all of $P$.
\end{lemma}

\begin{proof}
Let $x$ be an interior point of $P$ not seen by $\Pi$.
We show that there must be a boundary point that is also not seen by $\Pi$.
Consider the first segment $L$ of $\pi_E(x,\Pi)$.
It connects $x$ with a reflex vertex $v$ on $\partial P$. We extend $L$ away from $v$ to $x$ until it
hits $\partial P$ at $x'$. The point $x'$ is clearly not seen by $\Pi$ since if it were,
then $x$ would also be seen by $\Pi$.
\end{proof}

\begin{lemma}\label{Lemma_Two}
If a curved path $\Pi$ inside $P$ intersects all nonredundant cuts of $P$,
then for a nonredundant cut $C$ all vertices of $Pocket(C)$ must be visible from $\Pi$.
\end{lemma}

\begin{proof}
Let $v_C$ and $w_C$ be the endpoints of $C$, and $\overline{u_Cv_C}$ be the generator of $Pocket(C)$.
We show that $C$ is (at least partially) visible from every vertex $x$ of $Pocket(C)$.
Consider the paths $\pi_E(x,v_C)$ and $\pi_E(x,w_C)$.
These paths create a \emph{funnel} $F_x$ with \emph{apex} $a_x$ (see Figure~\ref{fig:Pocket}) \cite {Lee_1984}.
If $x$ is not visible from any point of $C$, then $a_x \neq x$.
In this case, $a_x$ is a reflex vertex of $Pocket(C)$ and the proof is continued as follows.
Let $b_x$ be the vertex adjacent to $a_x$ such that $a_x$ is the common apex between
$F_x$ and $F_{b_x}$ (if such $b_x$ does not exist, we define $b_x=x$).
Consider the cut $C'$ generated by $\overline{b_xa_x}$.
Since each of the funnel paths is \emph{inward} convex, $Pocket(C') \subset Pocket(C)$.
Thus, $C$ would be redundant which is a contradiction.
A similar argument proves that $u_C$ and $w_C$ are convex in $Pocket(C)$, and
all vertices of $Pocket(C)$ lie on the same side of $\overline{u_Cw_C}$.

\begin{figure}
	\centering
	\includegraphics[width=11cm,keepaspectratio=true]{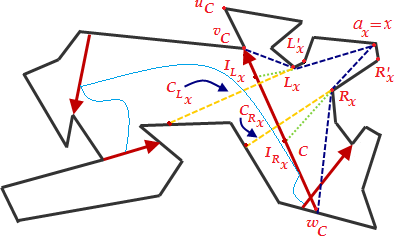}
	\caption{$Pocket(C)$ and its funnel for a vertex $x$, and path $\Pi$ (blue) intersecting all nonredundant cuts}
	\label{fig:Pocket}
\end{figure}

Consider the funnel $F_x$ with apex $a_x=x$.
Let $L_x \in F_x$ be the first vertex of $\pi_E(x,v_C)$ and
$R_x \in F_x$ be the first vertex of $\pi_E(x,w_C)$.
The extensions of $\overline{xL_x}$ and $\overline{xR_x}$
intersect $C$ at points $I_{L_x}$ and $I_{R_x}$, respectively
(if $I_{L_x}$ coincides with $I_{R_x}$, then $x$, $L_x$ and $R_x$ are collinear, and $a_x \neq x$).
Suppose that $L'_x$ is the vertex adjacent to $L_x$ such that
$L_x$ is the first vertex of $\pi_E(L'_x,v_C)$ ($L'_x$ can coincide with $x$).
Let $C_{L_x}$ be the cut generated by $\overline{L'_xL_x}$.
Similarly, $C_{R_x}$ is defined (see Figure~\ref{fig:Pocket}).

It is easy to see that
$C_{L_x}$ ($C_{R_x}$) always lies on the right (left)
of $\overline{L_xI_{L_x}}$ ($\overline{R_xI_{R_x}}$) or
it coincides with $\overline{L_xI_{L_x}}$ ($\overline{R_xI_{R_x}}$), i.e.,
$C_{L_x}$ and $C_{R_x}$ intersect $\overline{I_{L_x}w_C}$ and $\overline{I_{R_x}v_C}$, respectively
(note that in Figure~\ref{fig:Pocket}, the generators of $C_{L_x}$ and $C_{R_x}$ lie on the left and right sides of $x$, respectively).
With this property the following cases can occur:

\begin{enumerate}
\item [{$(a)$}] $\Pi$ intersects $\overline{I_{L_x}I_{R_x}}$ on $C$.
                Clearly, $x$ is visible from $\Pi$.
                
\item [{$(b)$}] $\Pi$ intersects $\overline{v_CI_{L_x}}$ on $C$.
                We will show that $\Pi$ must have an intersection with $C_{L_x}$
                and consequently with the extension of $\overline{xL_x}$.
                Thus, $x$ would be visible from $\Pi$.
                
\item [{$(c)$}] $\Pi$ intersects $\overline{I_{R_x}w_C}$ on $C$.
                We will show that $\Pi$ must have an intersection with $C_{R_x}$
                and consequently with the extension of $\overline{xR_x}$.
                Thus, $x$ would be visible from $\Pi$.
\end{enumerate}

In case $(b)$, if $C_{L_x}$ is nonredundant, then $\Pi$ must intersect it.
Otherwise, $Pocket(C_{L_x})$ contains at least one
$Pocket(C')$ for a nonredundant cut $C'$ and
$\Pi$ must go towards $C'$.
Therefore, $\Pi$ intersects $C_{L_x}$ (note that $P$ is a simple polygon).
The same result can be proved for the cut $C_{R_x}$ in case $(c)$.

For the \emph{degenerate} cases, i.e.,
if $x$ is visible from $v_C$ or $w_C$, we define $I_{L_x}=L_x=v_C$ or $I_{R_x}=R_x=w_C$, respectively,
only cases $(a)$, $(a)$ or $(b)$, or $(a)$ or $(c)$ can occur.
In these cases, we have at most one cut and a similar argument shows that
$x$ is visible from $\Pi$ (see Figure~\ref{fig:Degenerate} for three vertices $x$, $x'$ and $x''$).
\end{proof}

\begin{figure}
	\centering
	\includegraphics[width=7.5cm,keepaspectratio=true]{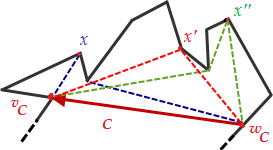}
	\caption{The degenerate cases for a nonredundant cut $C$}
	\label{fig:Degenerate}
\end{figure}

\begin{lemma}\label{Lemma_Three}
All vertices of $Q$ must be visible from $\Pi$, where $Q$ is the subpolygon of $P$ after cutting along all nonredundant cuts of $P$.
\end{lemma}

\begin{proof}
Since for every nonredundant cut $C$ of $P$ all vertices of $Pocket(C)$ lie on the same side of $C$
(proof of Lemma~\ref{Lemma_Two}),
such cutting is always possible
(If this cutting leads to an empty set, the proof is completed).
Clearly, $Q$ is a simple polygon.
If all vertices of $Q$ are convex ($Q$ is convex), they are visible from $\Pi$.
Thus, we suppose that $Q$ has at least one reflex vertex $v_i$.
Consider the two cuts $l_{i1}$ and $l_{i2}$ generated by $\overline{t_{i1}v_i}$ and $\overline{t_{i2}v_i}$,
respectively ($t_{i1}$ and $t_{i2}$ are the adjacent vertices of $v_i$).
Since $Q$ is a simple polygon and $\Pi$ is connected,
it is sufficient to show that $\Pi$ intersects $l_{ij}$ ($j=1$ or $2$),
or $t_{i1}$ and $t_{i2}$ are visible from $\Pi$ for all indices $i$ of reflex vertices of $Q$.

Let $S_P$ be the set of $Pocket(C)$ for every nonredundant cut $C$ of $P$.
Also, $S_P(l_{ij})$ is defined as the set of $P_k \in S_P$ such that
$P_k \subset Pocket(l_{ij})$ ($k=1,...,\lvert S_P \rvert$ and $j=1,2$).
For example, in Figure~\ref{fig:Internalreflex},
$S_P=\{P_1,P_2,P_3,P_4,P_5,P_6\}$,
while
$S_P(l_{q1})=\{P_1,P_2,P_3\}$, $S_P(l_{q2})=\{P_2,P_3,P_4,P_5,P_6\}$
and
$S_P(l_{r1})=\{P_1,P_2,P_3,P_4,P_6\}$, $S_P(l_{r2})=\{P_5\}$ for the reflex vertices $v_q$ and $v_r$, respectively.
Obviously, for every cut $C'$ of $Q$,
$Pocket(C')$ contains at least one member of $S_P$, i.e.,
all cuts of $Q$ are redundant in $P$.

\begin{figure}
	\centering
	\includegraphics[width=11.5cm,keepaspectratio=true]{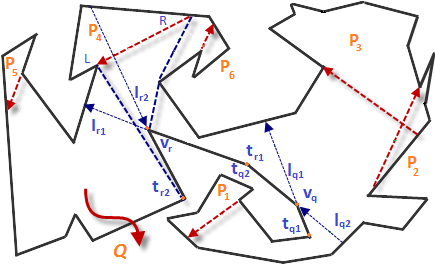}
	\caption{Cutting $P$ along all dashed nonredundant cuts yields $Q$;
	         $\pi_E(t_{r2},L)$ and $\pi_E(t_{r2},R)$ create the funnel $F_{t_{r2}}$}
	\label{fig:Internalreflex}
\end{figure}

If $S_P(l_{ij})=S_P$ ($j=1$ or $2$), then we show that $t_{ij}$ must be visible from $\Pi$.
Suppose that $t_{ij}$ is not visible from $\Pi$.
Let $z_j$ be the first vertex of $\pi_E(t_{ij},\Pi)$.
Clearly, $z_j$ would be a reflex vertex of $Q$ (if $z_j \in (P-Q)$, then $\Pi$ must be seen by $t_{ij}$).
We extend $\overline{t_{ij}z_j}$ away from $t_{ij}$ to $z_j$ until it hits $\partial P$ at $z'_j$.
The line segment $\overline{z_jz'_j}$ divides $P$ into two subpolygons, only one of which contains $t_{ij}$.
We define $T$ as the subpolygon containing $t_{ij}$.
Let $y_j$ be the adjacent vertex of $z_j$ in $T$ ($y_j \in Q$ can coincide with $t_{ij}$).
It is easy to see that there is no $P_k \in S_P$ such that $P_k \subset T$
(note that $S_P-S_P(l_{ij})=\emptyset$),
since otherwise, $\Pi$ intersects $\overline{z_jz'_j}$.
Let $C'$ be the cut of $Q$ generated by $\overline{y_jz_j}$.
Clearly, $Pocket(C') \subseteq T$.
Thus, $C'$ would be nonredundant in $P$, contradicting that all cuts of $Q$ are redundant in $P$.

Otherwise, $S_P(l_{ij}) \subset S_P$ ($j=1$ or $2$).
Since $l_{i1}$ and $l_{i2}$ are redundant cuts in $P$, there are $P_m \in S_P$ and $P_n \in S_P$ such that
$P_m \in S_P(l_{i1})$ and $P_n \in S_P(l_{i2})$ ($P_m$ and $P_n$ might coincide or not).

Suppose that $l_{ij}$ does not intersect any nonredundant cuts of $P$.
Therefore, there is a $P_o \in (S_P-S_P(l_{ij}))$, and
the portion of $\Pi$ that connects $P_m$ to $P_o$ or $P_n$ to $P_o$ must intersect $l_{ij}$
($P_m \neq P_o$ and $P_n \neq P_o$).

Now, suppose that $l_{ij}$ intersects a nonredundant cut $C$ of $P$.
Since $t_{ij}$ is visible from $C$, the same result in Lemma~\ref{Lemma_Two} can be proved
for the funnel $F_{t_{ij}}$ on the other side of $C$ ($x=t_{ij}$).
Thus, $t_{ij}$ would be visible from $\Pi$.
For example, in Figure~\ref{fig:Internalreflex},
there is the funnel $F_{t_{r2}}$ outside $P_4$.
\end{proof}

Therefore, all of $\partial P$ is visible from $\Pi$ (Lemma~\ref{Lemma_Two} and Lemma~\ref{Lemma_Three}) and according to
Lemma~\ref{Lemma_One},
$\Pi$ would be a visibility path.
This completes the proof of Theorem~\ref{Main}.
\end{proof}

\section{Future work}
It would be of interest to find a necessary and sufficient condition for other visibility models,
like conic visibility in simple polygons \cite {Garcia_2000}.
This is a generalization of the straight line visibility.
The anticipated result would have important applications in various areas such as:
the hidden surface removal problem in CAD/CAM, robot motion planning, and related computational geometry applications.

\section*{Acknowledgements}
We sincerely thank
Dr Jorg-Rudiger Sack from Carleton University
and Dr Ali Rajaei from Tarbiat Modares University, Computer Sciences group,
for their kind help and valuable comments which led to an improvement of this work.


\begin{thebibliography}{99}

\bibitem{Alsuwaiyel_1995}
M. H. Alsuwaiyel, D. T. Lee, Finding an approximate minimum-link visibility path inside a simple polygon,
\textit{Information Processing Letters},
\textbf{55}(2), (1995), 75-79.

\bibitem{Carlsson_1993}
S. Carlsson, H. Jonsson, B. J. Nilsson, Finding the shortest watchman route in a simple polygon,
\textit{Discrete and Computational Geometry},
\textbf{22}, (1999), 377-402.

\bibitem{Chin_1991}
W. Chin, S. C. Ntafos, Shortest watchman routes in simple polygons,
\textit{Discrete and Computational Geometry},
\textbf{6}, (1991), 9-31.

\bibitem{Das_1997}
G. Das, P. J. Heffernan, G. Narasimhan, LR-visibility in polygons,
\textit{Computational Geometry Theory and Applications},
\textbf{7}(1-2), (1997), 37-57.

\bibitem{Dror_2003}
M. Dror, A. Efrat, A. Lubiw, J.S.B. Mitchell, Touring a sequence of polygons,
\textit{In 35th Annual ACM Symposium on Theory of Computing}, (2003), 473-482.

\bibitem{Hammar_1997}
M. Hammar, B. J. Nilsson, Concerning the time bounds of existing shortest watchman route algorithms,
\textit{In International Symposium on Fundamentals of Computation Theory}, (1997), 210-221.

\bibitem{Garcia_2000}
J. García-López, P. A. Ramos, A unified approach to conic visibility,
\textit{Algorithmica},
\textbf{28}(3), (2000), 307-322.

\bibitem{Lee_1984}
D. T. Lee, F. P. Preparata, Euclidean shortest paths in the presence of rectilinear barriers,
\textit{Networks},
\textbf{14}(3), (1984), 393-410.

\bibitem{Nilsson_2001}
B. J. Nilsson, Approximating a shortest watchman route,
\textit{Fundamenta Informaticae},
\textbf{45}(3), (2001), 253-281.

\bibitem{Tan_2004}
X. Tan. Approximation algorithms for the watchman route and zookeeper’s problems,
\textit{Discrete Applied Mathematics},
\textbf{136}(2-3), (2004), 363-376.

\bibitem{Tan_2007}
X. Tan, A linear-time 2-approximation algorithm for the watchman route problem for simple polygons,
\textit{Theoretical Computer Science},
\textbf{384}(1), (2007), 92-103.

\bibitem{Tan_Hirata_1993}
X. Tan, T. Hirata, Constructing shortest watchman routes by divide and conquer,
\textit{International Symposium on Algorithms and Computation}, (1993), 68-77.

\bibitem{Tan_Hirata_Inagaki_1993}
X. Tan, T. Hirata, Y. Inagaki, An incremental algorithm for constructing shortest watchman routes,
\textit{International Journal of Computational Geometry and Applications},
\textbf{3}(04), (1993), 351-365.

\bibitem{Tan_Hirata_Inagaki_1999}
X. Tan, T. Hirata, Y. Inagaki, Corrigendum to "an incremental algorithm for constructing shortest watchman routes",
\textit{International Journal of Computational Geometry and Applications},
\textbf{9}(03), (1999), 319-323.

\bibitem{Zarrabi_2020}
M. R. Zarrabi, N. M. Charkari,
On approximations to minimum link visibility paths in simple polygons,
\textit{International Journal of Computer Mathematics: Computer Systems Theory},
\textbf{5}(4), (2020), 300-307.

\end{thebibliography}
\end{document}